\documentclass[12pt,english]{article}
\usepackage[english]{babel}
\usepackage[dvips]{graphicx}
\makeatletter
\usepackage{cite}
\usepackage{paralist}
\usepackage[cmex10]{amsmath}
\usepackage{amssymb}
\usepackage{amsthm}
\usepackage{psfrag}
\usepackage[dvips]{graphicx}
\usepackage{subfig}

\newtheorem{lemma}{Lemma}[section]

\newtheorem{theorem}{Theorem}[section]

\begin{document}

\def\ss{\sigma}
\def\fpzi{\frac{\partial}{\partial z_i}}
\def\fpzj{\frac{\partial}{\partial z_j}}
\def\BM{{\bf M}}
\def\BI{{\bf I}}
\def\BA{{\bf A}}
\def\BB{{\bf B}}
\def\E{{\bf E}}
\def\BD{{\bf D}}
\def\BT{{\bf T}}
\def\BX{{\bf X}}
\def\BU{{\bf U}}
\def\Bu{{\bf u}}
\def\BV{{\bf V}}
\def\Bv{{\bf v}}
\def\BI{\hbox{I}}
\def\BJ{\hbox{J}}
\def\det{\hbox{\rm det}}
\def\tr{\hbox{\rm tr}}
\def\fpx{\frac{\partial}{\partial x}}
\def\fpy{\frac{\partial}{\partial y}}
\def\det{\hbox{\rm det}}

\title{Optimizing Local Capacity of Wireless Ad Hoc Networks}

\author{
Salman Malik\footnote{INRIA Paris-Rocquencourt, France. Email: \texttt{salman.malik@inria.fr}}, 
Philippe Jacquet\footnote{INRIA Paris-Rocquencourt, France. Email: \texttt{philippe.jacquet@inria.fr}}
}
\date{}

\maketitle

\begin{abstract}

In this work, we evaluate local capacity of wireless ad hoc networks with several medium access protocols and identify the most optimal protocol. We define local capacity as the average information rate received by a receiver randomly located in the network. We analyzed grid pattern protocols where simultaneous transmitters are positioned in a regular grid pattern, pure ALOHA protocols where simultaneous transmitters are dispatched according to a uniform Poisson distribution and exclusion protocols where simultaneous transmitters are dispatched according to an exclusion rule such as node coloring and carrier sense protocols. Our analysis allows us to conjecture that local capacity is optimal when simultaneous transmitters are positioned in a grid pattern based on equilateral triangles and our results show that this optimal local capacity is at most double the local capacity of simple ALOHA protocol. Our results also show that node coloring and carrier sense protocols approach the optimal local capacity by an almost negligible difference. 

\end{abstract}

\section{Introduction}
\label{sec:intro}

Seminal work of Gupta \& Kumar~\cite{Gupta:Kumar} and later studies, {\it e.g.}, \cite{scaling,scaling2} quantify the capacity in wireless ad hoc networks in terms of scaling laws or bounds. These results are very important but may not provide detailed insight into the actual performance of various medium access protocols, such as the exact achievable capacity, or network design issues such as trade-offs involving protocol overhead versus performance of various medium access protocols, {\it etc}. Therefore, in order to get better insight into the designing of medium access protocols for wireless ad hoc networks, we will evaluate various protocols under the framework of {\em local capacity}.

Medium access protocols in wireless ad hoc networks can be broadly classified into two main classes: continuous time access and slotted access. In this article, we mainly focus on slotted medium access although many of our results can be applied to continuous time medium access. Within slotted medium access category, we distinguish node coloring, carrier sense multiple access (CSMA) and slotted ALOHA protocols. This article has two main goals. The first goal is to identify the most optimal medium access protocol in wireless ad hoc networks and evaluate its local capacity. Our second goal is to compare this optimal local capacity with the local capacities of above mentioned medium access protocols.

\section{General Settings}
\label{sec:model}

We consider a wireless ad hoc network where nodes are distributed uniformly over an infinite $2D$ map. In slotted medium access, at any given slot, simultaneous transmitters in the network are distributed like a set of points, \mbox{${\cal S}=\{z_1,z_2,\ldots,z_n,\ldots\}$}, where $z_i$ is the location of transmitter $i$. The spatial distribution of simultaneous transmitters, {\it i.e.} the set ${\cal S}$, depends on the medium access protocol employed by the nodes. Therefore, we do not adopt any {\em universal} model for the locations of simultaneous transmitters and assume that, in all slots, the set ${\cal S}$ has homogeneous density equal to $\lambda$. 

Let $\gamma_{ij}$ denote the channel gain from node $i$ to node $j$ such that the received power at node $j$ is $P_i\gamma_{ij}$, where $P_i$ is the transmit power of node $i$. We consider that all nodes use unit nominal transmit power. We ignore multi-path fading or shadowing effects and \mbox{$\gamma_{ij}=\vert z_i-z_j\vert^{-\alpha}$}, where \mbox{$\alpha>2$} is the attenuation coefficient and $\vert .\vert$ is the Euclidean norm of the vector. We also assume that the background noise power is negligible. Therefore, the transmission from node $i$ to node $j$ is successful only if the following condition is satisfied
$$
\frac{\vert  z_i-z_j\vert^{-\alpha}}{\sum_{k\neq i}\vert  z_k-z_j\vert^{-\alpha}}\geq K~,
$$
where $K$ is the minimum signal to interference ratio (SIR) threshold required for successfully receiving the packet. 

\section{Parameters of Interest}

The SIR of transmitter $i$ at any point $z$ on the plane is given by
\begin{equation}
S_i(z)=\frac{\vert  z-z_i\vert^{-\alpha}}{\sum_{j\neq i}\vert  z-z_j\vert^{-\alpha}}~.
\label{eq:sinr}
\end{equation} 

We call the reception area of transmitter $i$, the area of the plane, $A(z_i,\lambda,K,\alpha)$, where this transmitter is received with SIR at least equal to $K$. $A(z_i,\lambda,K,\alpha)$ also contains the point $z_{i}$ since here the SIR is infinite. The average size of $A(z_i,\lambda,K,\alpha)$ is $\sigma(\lambda,K,\alpha)$: $\sigma(\lambda,K,\alpha)=\E(|A(z_i,\lambda,K,\alpha)|)$, where $|A|$ is the size of an area $A$. 

Note that obviously, $\sigma(\lambda,K,\alpha)$ does not depend on $z_i$.

\subsection{Local Capacity}

Our principal parameter of interest is local capacity, hereafter referred to as capacity only, which is defined as the average information rate received by a receiver {\em randomly} located in the network. Consider a receiver at a random location $z$ in the network and let $N(z,K,\alpha)$ denote the number of reception areas it belongs to. Under general settings, following identity has been proved in \cite{Jacquet:2009}
\begin{equation}
\E(N(z,K,\alpha))=\lambda\sigma(\lambda,K,\alpha)~.
\label{eq:avg_no}
\end{equation}
$\E(N(z,K,\alpha))$ represents the average number of transmitters from which a receiver, randomly placed in the network, can receive with SIR at least equal to $K$. Under the hypothesis that a node can only receive at most one packet at a time, {\it e.g.}, when \mbox{$K>1$}, then \mbox{$N(z,K,\alpha)\le 1$}. The average information rate received by the receiver, $c(z,K,\alpha)$, is equal to $\E(N(z,K,\alpha))$ multiplied by the nominal capacity. Without loss of generality, we assume unit nominal capacity and we will compute 
\begin{equation}
c(z,K,\alpha)=\E(N(z,K,\alpha))=\lambda\sigma(\lambda,K,\alpha)~. 
\label{eq:poisson_hand_over_no}
\end{equation}

We will find exact bounds on capacity in wireless ad hoc networks with node coloring, CSMA and ALOHA protocols. We will also show that maximum capacity can be achieved with grid pattern protocols. Wireless networks of grid topologies are studied in, {\it e.g.}, \cite{Liu:Haenggi,Hong:Hua} and compared to networks with randomly distributed nodes. In contrast, we assume that only the simultaneous transmitters form a regular grid pattern. 

\subsection{Relationship of Local Capacity and Transport Capacity}

Gupta \& Kumar~\cite{Gupta:Kumar} introduced the concept of {\em transport capacity}. It is defined as the bit-meters that can be transported by the network per second. Their result is a scaling law, {\it i.e.}, the density of transport capacity scales as $Ck_1\sqrt{\lambda}$ bit-meters per second per unit area where $C$ is the nominal capacity and \mbox{$k_1>0$} depends on medium access protocol and system parameters. If all nodes are capable of transmitting at $C$ bits per second, the capacity of each node is $Ck_1/\sqrt{\lambda}$ bit-meters per second. It is also shown in~\cite{Gupta:Kumar} that under general settings, the effective radius of transmission is $k_2/\sqrt{\lambda}$ for some \mbox{$k_2>0$} which also depends on medium access protocol and system parameters. If each node transmits to a receiver which is randomly located within its effective radius of transmission or, in other words, its reception area, the information rate received by a receiver is constant and equal to $Ck_1/k_2$ bits per second. We evaluate the {\em average} of this information rate received by a receiver randomly located in the network, {\it i.e.}, the local capacity. Note that, this capacity also incorporates the pre-constants associated with the scaling law, {\it e.g.} $k_1$ and $k_2$, and it is independent of $\lambda$ as it is invariant for any {\em homothetic transformation} of the set of transmitters. 

\section{Related Works}
\label{sec:context}

In one of the first analyses on capacity of medium access protocols in wireless networks, \cite{Nelson:Kleinrock} studied slotted ALOHA and despite using a very simple geometric propagation model, the result is similar to what can be obtained under realistic SIR based interference model (non-fading, SIR threshold of $10.0$ and attenuation coefficient of $4.0$). Under a similar propagation model and assuming that all nodes are within range of each other, \cite{CSMA} evaluated CSMA protocol and compared it with slotted ALOHA in terms of throughput. \cite{Bartek} used simulations to analyze CSMA under a realistic SIR based interference model and compared it with ALOHA (slotted and un-slotted). For simulations, \cite{Bartek} assumed that transmitters send packets to their assigned receivers which are located at a fixed distance. 

\cite{Weber,Weber2} studied {\em transmission capacity}, which is the maximum number of successful transmissions per unit area at a specified outage probability, of ALOHA and code division multiple access (CDMA) protocols. They assumed that simultaneous transmitters form a homogeneous Poisson point process (PPP) and used the same model for location of receivers as in~\cite{Bartek}. The fact that the receivers are not a part of network (node distribution) model and are located at a fixed distance from the transmitters is a simplification. An accurate model of wireless networks should consider that the transmitters, transmit to receivers which are randomly located in their neighborhood. Similarly,~\cite{Weber3} analyzed transmission capacity of ALOHA and CSMA in networks with general node distributions under a restrictive hypothesis that density of interferers is very low and, asymptotically, approaches $0$. 

With exclusion protocols, like node coloring or CSMA, the correlation between the location of simultaneous transmitters makes it extremely difficult to develop a tractable analytical model. Some of the proposed approaches are as follows. \cite{Guard,Guard2} modeled interferers as PPP and exclude or suppress some of the interferers in the guard zone around a receiver. \cite{CSMA-Model,Weber3} used Mat\'ern point process however \cite{Busson} showed that it may lead to an underestimation of the density of simultaneous transmitters and proposed to use Simple Sequential Inhibition ({\em SSI}) or an extension of {\em SSI} called {\em SSI$_k$} point process. But, very few analytical results are available for $SSI$ or $SSI_k$ point processes and results are usually obtained via simulations.

In other related works, \cite{Haenggi} analyzed local (single-hop) throughput and capacity with slotted ALOHA, in networks with random and deterministic node placement, and TDMA, in $1D$ line-networks only. \cite{Zorzi2} determined the optimum transmission range under the assumption that interferers are distributed according to PPP whereas \cite{SR-ALOHA} gave a detailed analysis on the optimal probability of transmission for ALOHA which optimizes the product of simultaneously successful transmissions per unit of space by the average range of each transmission.

\section{Grid Pattern Based Protocols}
\label{sec:grid_pattern}

It can be argued that optimal capacity in wireless ad hoc networks can be achieved if simultaneous transmitters are positioned in a grid pattern. However, designing a protocol, which ensures that simultaneous transmitters are positioned in a grid pattern, is very difficult because of the limitations introduced by wave propagation characteristics and actual node distribution. For this, location aware nodes may be useful but the specification of a distributed protocol that would allow grid pattern transmissions is beyond the scope of this article. 

In this section, we will investigate the optimality of a grid pattern based protocol and later we will also present an analytical method to analyze its capacity. Grid pattern based protocols may have no practical implementation but their evaluation is interesting in order to establish an upper bound on the optimal capacity in wireless ad hoc networks.  

\subsection{Optimality of Grid Pattern Based Protocols}
\label{sec:grid_optimality}

In this section also, we consider that an infinite number of transmitters are uniformly distributed like a set of points, \mbox{${\cal S}=\{z_1,z_2,\ldots,z_n,\ldots\}$}, on an infinite $2D$ plane. The location of transmitter $i$ is denoted by $z_i$ and the center of the plane is at $(0,0)$. 

In order to simplify our analysis, we define a function $g_i(z)$ as 
$$
g_i(z)=\frac{|z-z_i|^{-\alpha}}{\sum_{j}|z-z_j|^{-\alpha}}~,
$$ 
where \mbox{$\alpha>2$}. The function $g_i(z)$ is similar to the SIR function $S_i(z)$, in (\ref{eq:sinr}), except that the summation in the denominator factor also includes the numerator factor. In order to simplify the notations, we will remove the reference to $z$ when no ambiguity is possible. 
We also define a function $f(g_i)$ which can be continuous or integrable. For instance, we will use \mbox{$f(g_i)=1_{g_i(z)\geq K'}$}, for some given $K'$ (in this case, the function is not continuous but we will not bother with this). In the following discussion, we can consider without loss of generality that the value of $K'$ is given by \mbox{$K'=\frac{K}{K+1}$}. Therefore, if transmitter $i$ is received successfully at location $z$ (with SIR at least equal to $K$, {\it i.e.}, \mbox{$S_i(z)\geq K$}), then \mbox{$g_i(z)\geq K'$} and $f(g_i)$ is equal to $1$.

We also assume a virtual disk on the plane centered at $(0,0)$ and of radius $R$. This allows us to express the density of set ${\cal S}$, $\nu({\cal S})$, in terms of the number of transmitters covered by the disk of radius $R$ or area $\pi R^2$, where $R$ approaches infinity, and it is given by a limit as
$$
\nu({\cal S})=\lim_{R\to\infty}\frac{1}{\pi R^2}\sum_i 1_{|z_i|\le R}~.
$$

We denote \mbox{$h(z)=\sum_i{f(g_i)}$}. Note that $h(z)$ is equal to the number of transmitters which can be successfully received at $z$ and its maximum value shall be $1$ if \mbox{$K>1$}. 

We define $\E(h(z))$ by the limit
$$
\E(h(z))=\lim_{R\to\infty}\frac{1}{\pi R^2}\int_{|z|\le R}h(z)dz^2~.
$$
The integration is over an infinite plane or, in other words, over the disk of radius $R$ where $R$ approaches infinity. The notations are simplified by taking $dxdy$ equal to $dz^2$. We denote the reception area of an arbitrary transmitter $i$ as 
$$
\ss_i=\int f(g_i)dz^2~,
$$
and we have 
$$
\E(h(z))=\lim_{R\to\infty}\frac{1}{\pi R^2}\sum_i 1_{|z_i|\le R}\ss_i=\nu({\cal S})\E(\ss_i)~,
$$
with
$$
\E(\ss_i)=\lim_{n\to\infty}\frac{1}{n}\sum_{i\le n}\ss_i~.
$$
As $R$ approaches infinity, $n$, {\it i.e.}, the number of transmitters in the set ${\cal S}$, covered by the disk of radius $R$, approaches infinity. 

Our objective is to optimize $\E(h(z))$ whose definition is equivalent to the definition of $\E(N(z,K,\alpha))$ and therefore capacity, $c(z,K,\alpha)$, in expressions (\ref{eq:avg_no}) and (\ref{eq:poisson_hand_over_no}) respectively. 

\subsubsection{First Order Differentiation}
\label{sec:grid_first_order}

We denote the operator of differentiation w.r.t. $z_i$ by $\nabla_i$. For \mbox{$i\neq j$}, we have 
$$
\nabla_i g_j=\alpha g_ig_j\frac{z-z_i}{|z-z_i|^2}
$$ 
and 
$$
\nabla_i g_i=\alpha(g_i^2-g_i)\frac{z-z_i}{|z-z_i|^2}~.
$$ 
Therefore
\begin{align}
\nabla_i h(z)&=\nabla_i\sum_if(g_i)=f'(g_i)\nabla_ig_i+\sum_{j\neq i}f'(g_j)\nabla_ig_j \notag \\
&=\alpha g_i\frac{z-z_i}{|z-z_i|^2}\Big(-f'(g_i)+\sum_j g_jf'(g_j)\Big)~. \notag
\end{align}
Although, we know that \mbox{$\int h(z)dz^2=\infty$}, we nevertheless have a finite $\nabla_i \int h(z)dz^2$. In other words, the sum $\sum_{j}\nabla_i\ss_j$ converges for all $i$. 

\begin{lemma}
For all $j$ in ${\cal S}$, \mbox{$\sum_{i}\nabla_i\ss_j=0$}. Indeed this would be the differentiation of $\ss_j$ when all points in ${\cal S}$ are translated by the same vector. Similarly, \mbox{$\sum_{i}\nabla_i\int  h(z)dz^2=0$}. 
\end{lemma}

\begin{theorem}
If the points in the set ${\cal S}$ are arranged in a grid pattern then: 
$$
\nabla_i \int h(z)dz^2=\sum_{j}\nabla_i\ss_j=0
$$ 
and grids patterns are {\em locally} optimal. 
\end{theorem}

\begin{proof}
If ${\cal S}$ is a set of points arranged in a grid pattern, then: \mbox{$\nabla_i \int h(z)dz^2=\sum_{j}\nabla_i\ss_j$} would be identical for all $i$ and, therefore, would be null since \mbox{$\sum_{i}\nabla_i\int  h(z)dz^2=0$}. 

We could erroneously conclude that,
\begin{compactitem}[-]
\item all grid sets are optimal and
\item all grid sets give the same $\E(h(z))$.
\end{compactitem}
In fact this is wrong: we could also conclude that $\E(\ss_i)$ does not vary but this will contradict that $\nu({\cal S})$ {\em must} vary. The reason of this error is that a grid set cannot be modified into another grid set with a {\em uniformly bounded transformation}, unless the two grid sets are just simply translated by a simple vector. 
\end{proof}

However, we prove that the grid sets are locally optimal within sets that can be uniformly transformed between each other. In order to cope with uniform transformation and to be able to transform a grid set to another grid set, we will introduce the linear group transformation. 

\subsubsection{Linear Group Transformation}

Here, we assume that the points in the plane are modified according to a continuous linear transform $M(t)$ where $\BM(t)$ is a matrix with \mbox{$\BM(0)=\BI$}, {\it e.g.}, \mbox{$\BM(t)=\BI+t\BA$} where $\BA$ is a matrix. 

Without loss of generality, we only consider $\ss_0$, {\it i.e.}, the reception area of the transmitter at $z_0$ which can be located anywhere on the plane. Under these assumptions, we have
$$
\frac{\partial}{\partial t}\ss_0=\sum_i (\BA z_i.\nabla_i\ss_0)=\tr\Big(\sum_{i} \BA^Tz_i\otimes \nabla_i\ss_0\Big)~.
$$

In other words, using the identity \mbox{$\frac{\partial \tr(\BA^T\BB)}{\partial \BA}=\BB$}, the derivative of $\ss_0$ w.r.t. matrix $\BA$ is exactly equal to \mbox{$\BD=\sum_{i} z_i\otimes \nabla_i\ss_0$}, 
such that
$$
\BD=\left[
\begin{array}{cc}
D_{xx}&D_{xy}\\
D_{yx}&D_{yy}
\end{array}
\right]
$$. 

Therefore, we can write the following identity
$$
\tr\Big(\BA^T \frac{\partial}{\partial \BA}\ss_0\Big)=\frac{\partial}{\partial t}\ss_0(t,\BA)\Big|_{t=0}~,
$$
where $\ss_0(t,\BA)$ is the transformation of $\ss_0$ under $M(t)$, {\it i.e.}, \mbox{$\ss_0(t,\BA)=\det(\BI+\BA t)\ss_0$}. We assume that {$\BM(t)=(1+t)\BI$} with \mbox{$\BA=\BI$}, {\it i.e.}, the linear transform is homothetic.

\begin{theorem}
$\BD$ is symmetric and $\tr(\BD)=2\ss_0$.
\end{theorem}

\begin{proof}
Under the given transform, $\ss_0(t,\BA)=\ss_0(t,\BI)=(1+t)^2\ss_0$. As a first property, we have \mbox{$\tr(\BD)=2\ss_0$}, since the derivative of $\ss_0$ w.r.t. identity matrix $\BI$ is exactly $2\sigma_0$, {\it i.e.}, 
$$
\tr(\BA^T\BD)=\tr(\BD)=\sigma_0'(0,\BI)=2\ss_0~.$$ 
The second property that $\BD$ is a symmetric matrix is not obvious. The easiest proof of this property is to consider the derivative of $\ss_0$ w.r.t. matrix 
\mbox{$
\BJ=\left[
\begin{array}{cc}
0&-1\\
1&0
\end{array}
\right]
$},
which is zero since $\BJ$ is the initial derivative for a rotation and reception area is invariant by rotation. 
Therefore, \mbox{$\tr(\BJ^T\BD)=D_{yx}-D_{xy}=0$}, which implies that $\BD$ is symmetric. 
\end{proof}

Note that $\BD$ can also be written in the following form
$$
\BD=\sum\limits_{i}z_i\otimes\nabla_i\ss_0=\int dz^{2}\sum\limits_{i}z_i\otimes\nabla_if(g_0)~.
$$
Let $\BT$ be defined as
$
\BT=\int dz^2\sum_i (z-z_i)\otimes\nabla_i f(g_0)~,
$
such that
$
\BD=\int \sum_{i} z\otimes \nabla_i f(g_0)dz^2 - \BT~.
$
The purpose of these definitions will become evident from theorems $3$ and $4$.

\begin{theorem}
We will show that $\int \sum_{i} z\otimes \nabla_i f(g_0)dz^2$ is equal to $\ss_0\BI$ and, therefore, $\BD=\ss_0\BI-\BT$. We will also prove that $\BT$ is symmetric. 
\end{theorem}

\begin{proof}
From the definition of $\BT$, we can see that the sum \mbox{$\sum_i (z-z_i)\otimes \nabla_i f(g_0)$} leads to a symmetric matrix since
\begin{eqnarray*}
\BT&=&\alpha\int f'(g_0)\Big(\frac{g_0^2-g_0}{|z-z_0|^2}(z-z_0)\otimes(z-z_0)+\\
&&\sum_{i\neq 0}\frac{g_0g_i}{|z-z_i|^2}(z-z_i)\otimes(z-z_i)\Big)dz^2~,
\end{eqnarray*}
and the left hand side is made of \mbox{$(z-z_i)\otimes(z-z_i)$} which are symmetric matrices. This implies that $\BT$ is also symmetric.

We can see that \mbox{$\sum_{i}\nabla_i f(g_0)=-\nabla f(g_0)$}, and using integration by parts we have
\begin{eqnarray*}
\lefteqn{\int \sum_i z\otimes \nabla_i f(g_0)dz^2=-1\times}\\
&\Biggl[
\begin{array}{cc}
\int x \fpx f(g_0)dxdy&\int x \fpy f(g_0)dxdy\\
\int y \fpx f(g_0)dxdy&\int y \fpy f(g_0)dxdy
\end{array}
\Biggr]=
\Biggl[\begin{array}{cc}
\ss_0&0\\
0&\ss_0
\end{array}\Biggr]~,
\end{eqnarray*}
which is symmetric and equal to $\ss_0\BI$. The sum/difference of symmetric matrices is also a symmetric matrix and, therefore, $\BD$ is a symmetric matrix and $\BD=\ss_0\BI-\BT$. 
\end{proof}


Now, we will only consider grid patterns and, by virtue of a grid pattern, we can have 
$$\E(\ss_i)=\ss_0=\int f(g_0)dz^2~,
$$ 
and \mbox{$\E(h(z))=\nu({\cal S})\ss_0$}. Under homothetic transformation, $\nu({\cal S})$ and $\ss_0$ are transformed but $\nu({\cal S})\ss_0$ remains invariant.

\begin{theorem}
If the pattern of the points in set ${\cal S}$ is optimal w.r.t. linear transformation of the set, $\BD=\ss_0\BI$ and $\BT=0$.
\end{theorem}

\begin{proof}
The derivative of $\ss_0$ w.r.t. matrix $\BA$ is exactly equal to $\BD$. 
Similarly, under the same transformation
$$
\frac{\partial}{\partial t}\nu({\cal S})=\frac{1}{\det(\BI+\BA t)}\nu({\cal S})~,
$$
and for $\BA=\BI$, it can be written as $\nu'({\cal S})(t,\BI)=\nu({\cal S})/(1+t)^2$.

In any case, the derivative of $\nu({\cal S})$ w.r.t. matrix $\BA$ is exactly equal to $-\BI\nu({\cal S})$. 
We also know that if the pattern is optimal w.r.t. linear transformation, 
the derivative of $\nu({\cal S})\ss_0$ w.r.t. to matrix $\BA$ shall be null. This implies that
$
\nu({\cal S})\BD-\BI\nu({\cal S})\ss_0=0~,
$
which leads to $\BD=\ss_0\BI$ and $\BT=0$.
\end{proof}

\begin{figure}[!t]
\centering
\psfrag{a}{$\sqrt{3}d$}
\psfrag{b}{$2d$}
\psfrag{c}{$d$}
\includegraphics[scale=0.65]{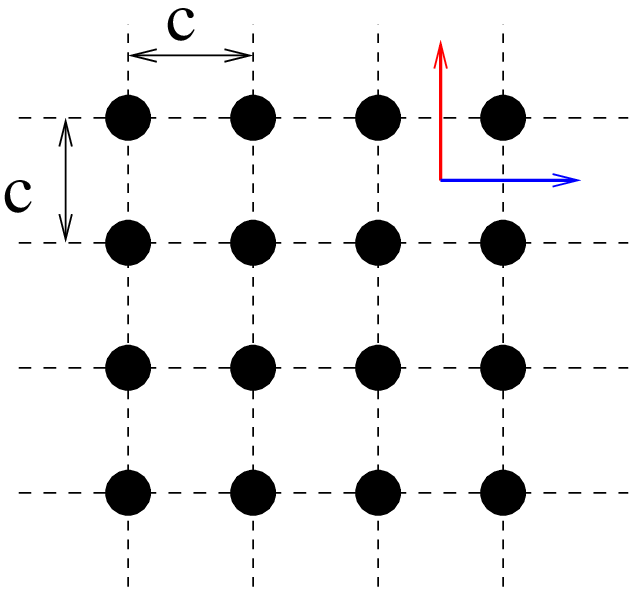}
\includegraphics[scale=0.65]{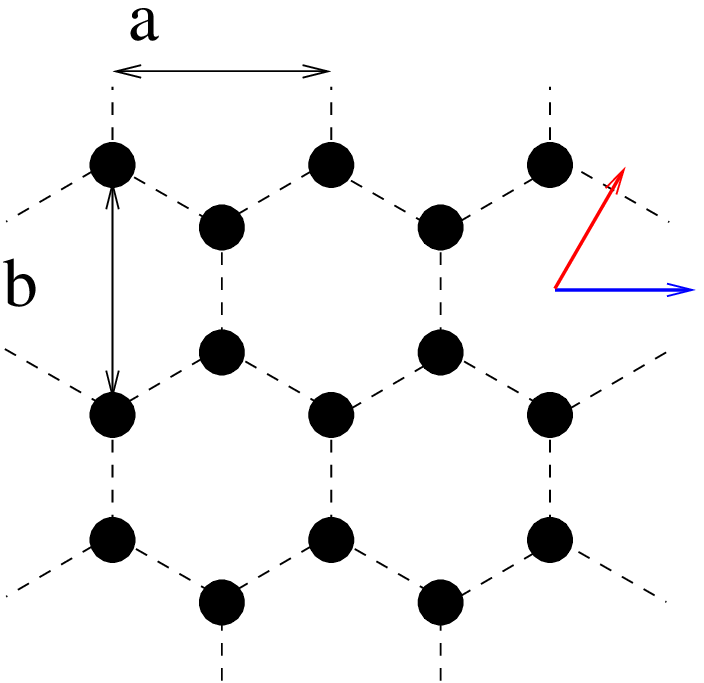}
\includegraphics[scale=0.65]{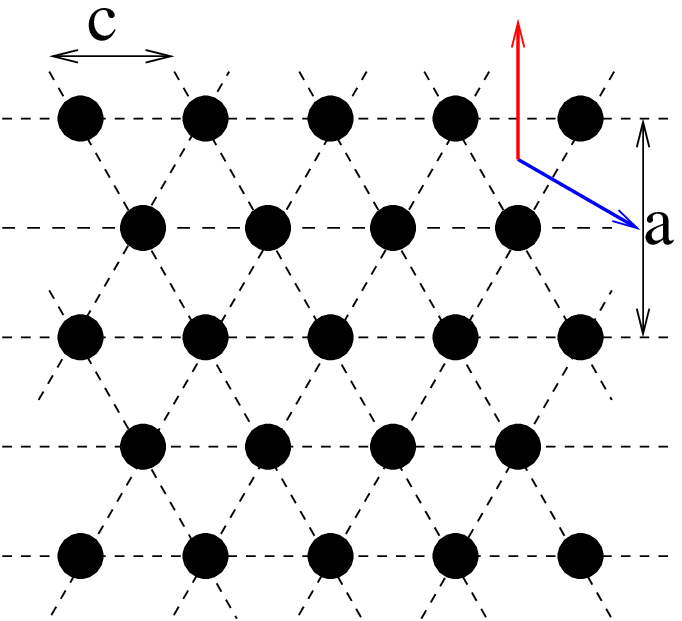}
\caption{Square, Hexagonal and Triangular grids. The arrows (blue and red) represent the invariance of Eigen values w.r.t. isometric symmetries of the grids.}
\label{fig:grid_layouts}
\end{figure}

We know that $\BT$ is symmetric and $\BT=0$. Thus, \mbox{$\tr(\BT)=0$}, {\it i.e.}, Eigen values are invariant by rotation. When a grid is optimal, we must have \mbox{$\BT=0$}. In any case, the matrix $\BT$ must be invariant w.r.t. isometric symmetries of the grid. On $2D$ plane, the grid patterns which satisfy this condition are square, hexagonal and triangular grids. The square grid is symmetric w.r.t. any horizontal or vertical axes of the grid and, in particular, with rotation of $\pi/2$ represented by $\BJ$. Therefore, the {\em Eigen system} must be invariant by rotation of $\pi/2$. This implies that the {\em Eigen values} are the same and therefore null since \mbox{$\tr(\BT)=0$}. Same argument also applies for the hexagonal grid with the invariance for $\pi/3$ rotation and for the triangular pattern with invariance for $2\pi/3$ rotation. 


\subsection{Reception Areas}
\label{sec:rx_area_2}

\begin{figure}[!t]
\centering
\psfrag{a}{$z_i$}
\psfrag{b}{$z$}
\psfrag{c}{$dz=J\frac{\nabla S_{i}(z)}{|\nabla S_{i}(z)|}\delta t$}
\psfrag{d}{${\cal C}(z_i,K,\alpha)$}
\psfrag{e}{$A(z_i,\lambda,K,\alpha)$}
\includegraphics[scale=.85]{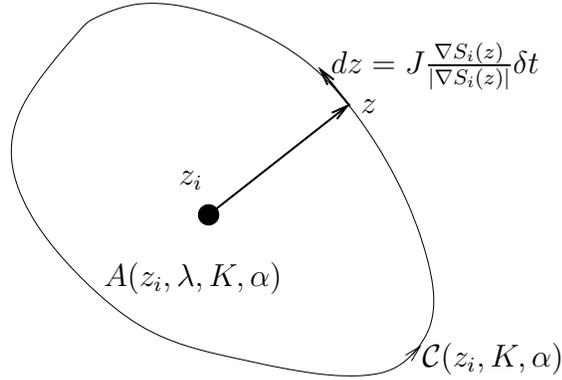}
\caption{Computation of the reception area of transmitter $i$.\label{fig:snr_gradient}}
\end{figure}

Here, the simultaneous transmitters, {\it i.e.}, the set ${\cal S}$ is a set of points arranged in a grid pattern. We consider that, for every slot, the grid pattern is the same {\em modulo} a translation. We have covered grid layouts of square, hexagonal and triangle as shown in Fig. \ref{fig:grid_layouts}. Grids are constructed from $d$ which defines the minimum distance between neighboring transmitters and can be derived from the hop-distance parameter of a typical TDMA-based protocol. The density of grid points, $\lambda$, depends on $d$. However, the capacity, $c(z,K,\alpha)$, is independent of the value of $d$ or, for that matter, $\lambda$ as it is invariant for any homothetic transformation of the set of transmitters. 

Our aim is to compute the size of the reception area, $A(z_i,\lambda,K,\alpha)$, around each transmitter $i$. By consequence of the regular grid pattern, all reception areas are the same {\em modulo} a translation (and a rotation for the hexagonal pattern), and their surface area size, $\sigma(\lambda,K,\alpha)$, is the same.  

If ${\cal C}(z_i,K,\alpha)$ is the closed curve that forms the boundary of $A(z_i,\lambda,K,\alpha)$ and $z$ is a point on ${\cal C}(z_i,K,\alpha)$, we have
\begin{equation}
\sigma(\lambda, K,\alpha)=\frac{1}{2}\displaystyle\int\limits_{{\cal C}(z_i,K,\alpha)}\det(z-z_{i},dz)~,
\label{eq:area_integral}
\end{equation}
where $\det(a,b)$ is the determinant of vectors $a$ and $b$ and $dz$ is the vector tangent to ${\cal C}(z_i,K,\alpha)$ at point $z$. \mbox{$\det(z-z_i,dz)$} is the cross product of vectors $(z-z_{i})$ and $dz$ and gives the area of the parallelogram formed by these two vectors. 

The SIR $S_{i}(z)$ of transmitter $i$ at point $z$ is given by (\ref{eq:sinr}).
We assume that at point $z$, $S_{i}(z)=K$. On point $z$ we can also define the gradient of $S_{i}(z)$, $\nabla S_{i}(z)$.
$\nabla S_i(z)$ is inward normal to the curve ${\cal C}(z_i,K,\alpha)$ and points towards $z_i$. The vector $dz$ is co-linear with $J\frac{\nabla S_{i}(z)}{|\nabla S_{i}(z)|}$ where $J$ is the anti-clockwise rotation of $3\pi/2$ (or clockwise rotation of $\pi/2$) given by
$$
J=\left[\begin{array}{cc}
0 & 1\\
-1 & 0\end{array}\right]~.
$$ 
Therefore, we can fix $dz=J\frac{\nabla S_{i}(z)}{|\nabla S_{i}(z)|}\delta t$ and in (\ref{eq:area_integral}) 
\begin{align*}
\det(z-z_{i},dz)&=(z-z_i)\times J\frac{\nabla S_{i}(z)}{|\nabla S_{i}(z)|}\delta t\\
&=-(z-z_{i}).\frac{\nabla S_{i}(z)}{|\nabla S_{i}(z)|}\delta t~,
\end{align*}
where $\delta t$ is assumed to be infinitesimally small. The sequence of points $z(k)$ computed as 
\begin{align*}
z(0) & =z\\
z(k+1) & =z(k)+J\frac{\nabla S_{i}(z(k))}{|\nabla S_{i}(z(k))|}\delta t~,\end{align*} gives a discretized and numerically convergent parametric representation of ${\cal C}(z_i,K,\alpha)$ by finite elements. 

Therefore, (\ref{eq:area_integral}) reduces to
\begin{equation}
\sigma(\lambda, K,\alpha)\approx-\frac{1}{2}\sum_{k}(z(k)-z_{i}).\frac{\nabla S_{i}(z(k))}{|\nabla S_{i}(z(k))|}\delta t~,
\label{eq:area_integral_2}
\end{equation}
assuming that we stop the sequence $z(k)$ when it loops back on or close to the point $z$. 

The point, \mbox{$z(0)=z$}, can be found using Newton's method. First approximate value of $z$, required by Newton's method, can be computed assuming only one interferer nearest to the transmitter $i$. The negative sign in (\ref{eq:area_integral_2}) is automatically negated by the dot product of vectors $(z(k)-z_{i})$ and $\nabla S_{i}(z(k))$.

\subsection{Capacity}

$$
c(z,K,\alpha)=\E(N(z,K,\alpha))=N(z,K,\alpha)=\lambda\sigma(\lambda,K,\alpha)~,
$$ 
where $\sigma(\lambda,K,\alpha)$ is computed using the above described method. 

\section{Protocols Based on Exclusion Rules}
\label{sec:practical}

Because of the lack of any satisfactory and tractable analytical model for protocols based on exclusion rules, like node coloring and CSMA based protocols, we will use Monte Carlo simulations along with the analytical method of \S \ref{sec:rx_area_2} to compute their capacity. Here, we will discuss the models of these protocols which we will employ in our Monte Carlo simulations. For the following discussion, the set of all nodes in the network is ${\cal N}$. In practical implementation, this set is finite but in theory, it can be infinite but with a uniform density.

\subsection{Node Coloring Based Protocols}
\label{sec:tdma}

Node coloring protocols use a managed transmission scheme based on time division multiple access (TDMA) approach. The aim is to minimize the interference between transmissions that cause packet loss. These protocols assign colors to nodes that correspond to periodic slots, {\it i.e.}, nodes that satisfy a spatial condition, either based on physical distance or distance in terms of number of hops, will be assigned different colors. For example, in order to avoid collisions at receivers, all nodes within $k$ hops are assigned unique colors. Typical value of $k$  is $2$. A few practical implementations of node coloring protocols are ~\cite{unified,NAMA,SEEDEX,FPRP,DRAND}. 

Instead of considering any particular protocol, we will present a model which ensures that transmitters use an exclusion distance in order to avoid the use of same slot within a certain distance. This exclusion distance is defined in terms of euclidean distance $d$ which may be derived from the distance parameter of a typical TDMA-based protocol. Therefore, a slot cannot be shared within a distance of $d$ or, in other words, nodes transmitting in the same slot shall be located at a distance greater or equal to $d$ from each other. 

Following is a model of node coloring protocols which constructs the set of simultaneous transmitters, ${\cal S}$, in each slot (this is supposed to be done off-line so that transmission patterns periodically recur in each slot).
\begin{compactenum}
\item Initialize ${\cal M}={\cal N}$ and ${\cal S}=\emptyset$.
\item Randomly select a node $s_{i}$ from ${\cal M}$ and add it to the set ${\cal S}$, i.e, ${\cal S}={\cal S}\cup\{s_{i}\}$. Remove $s_{i}$ from the set ${\cal M}$.
\item Remove all nodes from the set ${\cal M}$ which are at distance less than $d$ from $s_{i}$.
\item If set ${\cal M}$ is non-empty, repeat from step $2$.
\end{compactenum}
This model maximizes the number of simultaneous transmitters in each slot and should give the maximum capacity achievable with any node coloring protocol which {\em may not} prioritize the nodes for coloring, {\it e.g.}, \cite{DRAND}.

\subsection{CSMA Based Protocols}
\label{sec:csma}

Extremely managed transmission scheduling in node coloring protocols has significant overhead, {\it e.g.}, because of the control traffic or message passing required to achieve the distributed algorithms that resolve color assignment conflicts. CSMA based protocols are simpler but are more demanding on the physical layer. Before transmitting on the channel, a node verifies if the medium is idle by sensing the signal level. If the detected signal level is below a certain threshold, medium is assumed idle and the node transmits its packet. Otherwise, it may invoke a random back-off mechanism and wait before attempting a retransmission. CSMA/CD (CSMA with collision detection) and CSMA/CA (CSMA with collision avoidance), which is also used in IEEE 802.11, are the modifications of CSMA for performance improvement. 

We will adopt a model of CSMA based protocol where nodes contend to access medium at the beginning of each slot. In other words, nodes transmit only after detecting that medium is idle. We assume that nodes defer their transmission by a tiny back-off time, from the beginning of a slot, and abort their transmission if they detect that medium is not idle. We also suppose that detection time and receive to transmit transition times are negligible and, in order to avoid collisions, nodes use randomly selected (but different) back-off times. Therefore, the main effect of back-off times is in the production of a random order of the nodes in competition.  

For the evaluation of the performance of CSMA based protocols, we will use the following simplified construction of the set of simultaneous transmitters ${\cal S}$. 
\begin{compactenum}
\item Initialize ${\cal M}={\cal N}$ and ${\cal S}=\emptyset$.
\item Randomly select a node $s_{i}$ from ${\cal M}$ and add it to the set ${\cal S}$, {\it i.e.}, ${\cal S}={\cal S}\cup\{s_{i}\}$. Remove $s_{i}$ from the set ${\cal M}$.
\item Remove all nodes from the set ${\cal M}$ which can detect a combined interference signal of power higher than $\theta$ (carrier sense threshold), from all transmitters in the set ${\cal S}$, {\it i.e.}, if 
$$
\sum_{s_i\in\cal{S}}|z_i-z_j|^{-\alpha}\geq \theta~, 
$$ 
remove $s_j$ from $\cal{M}$. Here, $z_i$ is the position of $s_i$ and $|z_i-z_j|$ is the euclidean distance between $s_i$ and $s_j$.
\item If set ${\cal M}$ is non-empty, repeat from step $2$.
\end{compactenum}
These steps model a CSMA based protocol which requires that transmitters do not detect an interference of signal level equal to or higher than $\theta$, during their back-off periods, before transmitting on the medium. At the end of the construction of set ${\cal S}$, some transmitters may experience interference of signal level higher than $\theta$. However, this behavior is in compliance with a realistic CSMA based protocol where nodes, which started their transmissions, or, in other words, are already added to the set ${\cal S}$ do not consider the increase in signal level of interference resulting from later transmitters.

\subsection{Reception Areas} 

The average size of the reception area of an arbitrary transmitter is evaluated via Monte Carlo simulation using the analytical method of \S \ref{sec:rx_area_2}. The value of $d$, in case of node coloring protocol, or $\theta$, in case of CSMA based protocol, can be tuned to obtain an average transmitter density of $\lambda$. 

\subsection{Capacity}

$$
c(z,K,\alpha)=\E(N(z,K,\alpha))=\lambda \sigma(\lambda,K,\alpha)~,
$$ 
is also computed via Monte Carlo simulation. The capacity, $c(z,K,\alpha)$, is invariant for any homothetic transformation of $\lambda$ and, therefore, it is also independent of the values of protocol parameters $\theta$ or $d$.

\section{Slotted ALOHA Protocol}
\label{sec:aloha}

In slotted ALOHA protocol, nodes do not use any complicated managed transmission scheduling and transmit their packets independently (with a certain medium access probability), {\it i.e.}, in each slot, each node decides independently whether to transmit or otherwise remain silent. Therefore, the set of simultaneous transmitters, in each slot, can be given by a uniform Poisson distribution of mean equal to $\lambda$ transmitters per unit square area~\cite{Jacquet:2009,SR-ALOHA,Weber2}. Here, we will make use of the results from \cite{Jacquet:2009} to derive the analytical expression for the capacity with slotted ALOHA protocol.

\subsection{Reception Areas}

Under the given settings, the average size of the reception area around an arbitrary transmitter satisfies the identity
\begin{equation}
\sigma(\lambda,K,\alpha)=\frac{1}{\lambda}\frac{\sin(\frac{2}{\alpha}\pi)}{\frac{2}{\alpha}\pi}K^{-\frac{2}{\alpha}}~.
\label{eq:poisson_area}
\end{equation}

We notice that when $\alpha$ approaches infinity, $\sigma(\lambda,K,\infty)$ approaches $1/\lambda$. This is due to the fact that when $\alpha$ is very large, all nodes other than the closest transmitter tend to contribute as a negligible source of interference and consequently the reception areas turn to be the Voronoi cells around every transmitter. This holds for all values of $K$. The average size of Voronoi cell being equal to the inverse density of the transmitters,
$1/\lambda$, we get the asymptotic result. 

\subsection{Capacity}

In this case, the analytical expressions (\ref{eq:poisson_hand_over_no}) and (\ref{eq:poisson_area}) lead to 
\begin{equation}
c(z,K,\alpha)=\E(N(z,K,\lambda))=\sigma(1,K,\alpha)~.
\label{eq:poisson_capacity}
\end{equation}

\section{Evaluation and Results}
\label{sec:simulations}

In order to approach an infinite map, we perform numerical simulations in a very large network spread over $2D$ square map with length of each side equal to $10000$ meters.

\subsection{Grid Pattern Based Protocols}

In this case, transmitters are spread over this network area in square, hexagonal or triangular pattern. For all grid patterns, we set $d$ equal to $25$ meters although it will have no effect on the validity of our conclusions as capacity, $c(z,K,\alpha)$, is independent of $\lambda$. To keep away edge effects, we compute the size of the reception area of transmitter $i$, located in the center of the network area: \mbox{$z_{i}=(x_{i},y_{i})=(0,0)$}. The network area is large enough so that the reception area of transmitter $i$ is close to its reception area in an infinite map. $\lambda$ depends on the type of grid and it is computed from the total number of transmitters spreading over the network area.

\subsection{Protocols Based on Exclusion Rules}

\subsubsection{Node Coloring Based Protocols}

Performance of node coloring based protocols is analyzed, via simulations, using the model specified in \S \ref{sec:tdma}. We set $d$ equal to $25$ meters. 

\subsubsection{CSMA Based Protocols}

In order to evaluate the capacity of CSMA based protocols, we perform simulations using the model specified in \S \ref{sec:csma}. The value of carrier sense threshold, $\theta$, is set equal to $1\times 10^{-5}$.

\subsubsection{Simulations}

We consider that nodes are uniformly distributed over the network area and simultaneous transmitters, in each slot, are selected according to the model of each medium access protocol. Considering the practical limitations introduced by the bounded network area, we use the following Monte Carlo method to evaluate $\sigma(\lambda,K,\alpha)$. We {\em only} compute the size of the reception area of a transmitter located nearest to the center of the network area and $\sigma(\lambda,K,\alpha)$ is the average of results obtained with $10000$ samples of node distributions. Similarly, $\lambda$ is also the average of the density of simultaneous transmitters obtained with these $10000$ samples of node distributions. Note that the protocol models select simultaneous transmitters randomly and transmitters are uniformly distributed over the network area. Therefore, using Monte Carlo method, {\it i.e.}, a large number of samples of node distributions and, with each sample, only measuring the reception area of a transmitter located nearest to the center of the network area gives an accurate approximation of $\sigma(\lambda,K,\alpha)$ in an infinite map with given values of $d$ or $\theta$. 

It can be argued that, in case of CSMA, density of simultaneous transmitters is higher on the boundaries of the network area, because of lower signal level of interference, as compared to the central region. The network area is very large and we observed that the difference, in spatial density of simultaneous transmitters, on the boundaries and central region is negligible. We also know that the capacity, $c(z,K,\alpha)$, is independent of $\lambda$ which depends on $d$ or $\theta$. However, if the node density is very low, it will also have an impact on the packing (density) of simultaneous transmitters in the network. In fact, $\lambda$ should be maximized to the point where no additional transmitter can be added to the network under given values of $d$ or $\theta$. This can be achieved by keeping the node density very high, {\it e.g.}, we observed that the node density of $1$ node per square meter is sufficient and further increasing the node density does not increase $\lambda$. In order to keep away the edge effects, values of $d$ or $\theta$ are chosen such that $\lambda$ is sufficiently high and edge effects have minimal effect on the central region of the network.

\subsection{Slotted ALOHA Protocol}

In case of slotted ALOHA protocol, capacity, $c(z,K,\alpha)$, is computed from analytic expressions (\ref{eq:poisson_area}) and (\ref{eq:poisson_capacity}).

\begin{figure*}[!t]
\centering
\subfloat[$K$ is varying and $\alpha$ is fixed at $4.0$.]{
	\hspace{-0.8 cm}
	\includegraphics[scale=1]{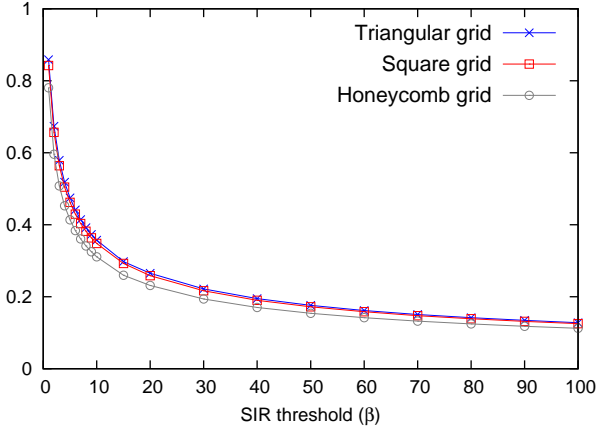}
}
\subfloat[$K$ is fixed at $10.0$ and $\alpha$ is varying.]{
	\hspace{-0.8 cm}
	\includegraphics[scale=1]{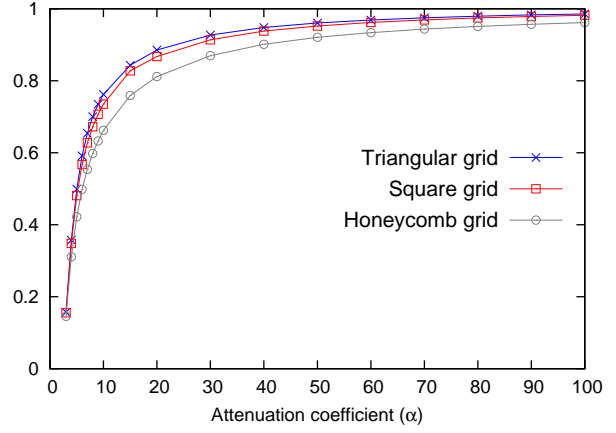}
}
\caption{Capacity, $c(z,K,\alpha)$, of grid pattern (triangular, square and hexagonal) based protocols.
\label{fig:comparison1}}
\end{figure*}

\begin{figure*}[!t]
\centering
\subfloat[$K$ is varying and $\alpha$ is fixed at $4.0$.]{
	\hspace{-0.8 cm}
	\includegraphics[scale=1]{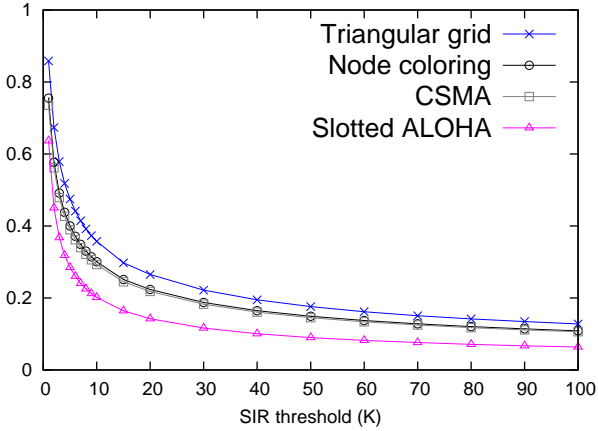}
}
\subfloat[$K$ is fixed at $10.0$ and $\alpha$ is varying.]{
	\hspace{-0.8 cm}
	\includegraphics[scale=1]{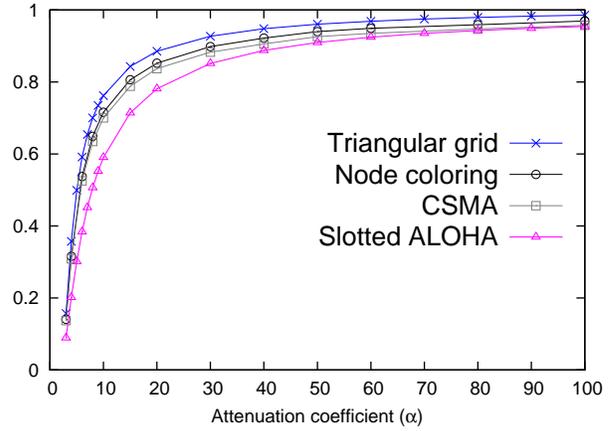}
}
\caption{Capacity, $c(z,K,\alpha)$, of triangular grid (from Fig. \ref{fig:comparison1}), node coloring, CSMA and slotted ALOHA protocols.
\label{fig:comparison2}}
\end{figure*}

\subsection{Observations}

The values of SIR threshold, $K$, and attenuation coefficient, $\alpha$, depend on the underlying physical layer or system parameters and are usually fixed and beyond the control of network/protocol designers. However, to give the reader an understanding of the influence of these parameters on the capacity, $c(z,K,\alpha)$, of different medium access protocols, we assume that these parameters are variable. Figures \ref{fig:comparison1}(a) and \ref{fig:comparison2}(a) show the comparison of capacity, $c(z,K,\alpha)$, with grid patterns, node coloring, CSMA and slotted ALOHA protocols with $K$ varying and \mbox{$\alpha=4.0$}. Similarly, Fig. \ref{fig:comparison1}(b) and \ref{fig:comparison2}(b) show the comparison of these protocols with \mbox{$K=10.0$} and $\alpha$ varying. We know that as $\alpha$ approaches infinity, reception area around each transmitter turns to be a Voronoi cell with an average size equal to $1/\lambda$. Therefore, as $\alpha$ approaches infinity, $c(z,K,\alpha)$ approaches 1. For slotted ALOHA protocol, (\ref{eq:poisson_area}) and (\ref{eq:poisson_capacity}) also arrive at the same result. For other protocols, we computed $c(z,K,\alpha)$ with $\alpha$ increasing up to $100$ and from the results, we can observe that asymptotically, as $\alpha$ approaches infinity, $c(z,K,\alpha)$ approaching 1 is true for all protocols. 

From the results, we can see that the maximum capacity in wireless ad hoc networks can be obtained with triangular grid pattern based protocol. In order to quantify the improvement in capacity by triangular grid pattern protocol over other protocols, we perform a scaled comparison of triangular grid pattern, slotted ALOHA, node coloring and CSMA based protocols which is obtained by dividing the capacity, $c(z,K,\alpha)$, of all these protocols with the capacity, $c(z,K,\alpha)$, of triangular grid pattern protocol. Figure \ref{fig:improvement} shows the scaled comparison with $K$ and $\alpha$ varying. It can be observed that triangular grid pattern protocol can achieve, {\it at most}, double the capacity of simple slotted ALOHA protocol whereas node coloring and CSMA based protocols can achieve almost \mbox{$85\sim90\%$} of the optimal capacity obtained with triangular grid pattern protocol. 

Triangular grid pattern can be visualized as an optimal node coloring which ensures that transmitters are exactly at distance $d$ from each other whereas, in case of random node coloring, transmitters are selected randomly and only condition is that they must be at a distance greater or equal to $d$ from each other. The exclusion region around each transmitter is a circular disk of radius $d/2$ with transmitter at the center. The triangular grid pattern can achieve a packing density of \mbox{$\pi/\sqrt{12}\approx0.9069$}. The packing density is defined as the proportion of network area covered by the disks of simultaneous transmitters. However, random packing of disks, which is the case in random node coloring, can achieve a packing density in the range of \mbox{$0.54\sim0.56$} only~\cite{disk,Busson}. We have seen in the results that even this sub-optimal packing of simultaneous transmitters by random node coloring can achieve almost similar capacity as obtained with optimal packing by triangular grid pattern. 

We observe that capacity with CSMA is slightly lower (by approximately $3\%$) as compared to node coloring and this is irrespective of the value of carrier sense threshold. The reason of slightly lower capacity with CSMA is that exclusion rule is based on carrier sense threshold, rather than the distance in-between simultaneous transmitters, which may not allow to pack more transmitters, in each slot, that would have been possible with node coloring protocols. In other words, CSMA may result in a lower packing density of simultaneous transmitters as compared to node coloring protocol. This can also be observed by comparing the densities of {\em SSI} and {\em SSI$_k$} point processes in \cite{Busson} and also explains the slightly lower capacity of CSMA as compared to node coloring protocol. However, as $\alpha$ approaches infinity, $\lambda$ with CSMA approaches the node density and reception area around each transmitter also becomes a Voronoi cell with an average size equal to the inverse of node density. In fact, asymptotically, as $\alpha$ approaches infinity, capacity, $c(z,K,\alpha)$, approaches $1$.

\begin{figure*}[!t]
\centering
\subfloat[$K$ is varying and $\alpha$ is fixed at $4.0$.]{
	\hspace{-0.8 cm}
	\includegraphics[scale=1]{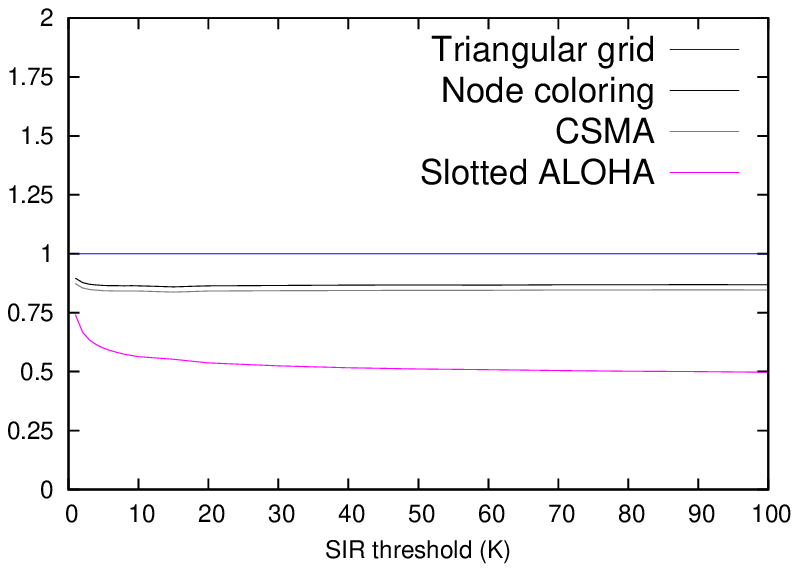}
}
\subfloat[$K$ is fixed at $10.0$ and $\alpha$ is varying.]{
	\hspace{-0.8 cm}
	\includegraphics[scale=1]{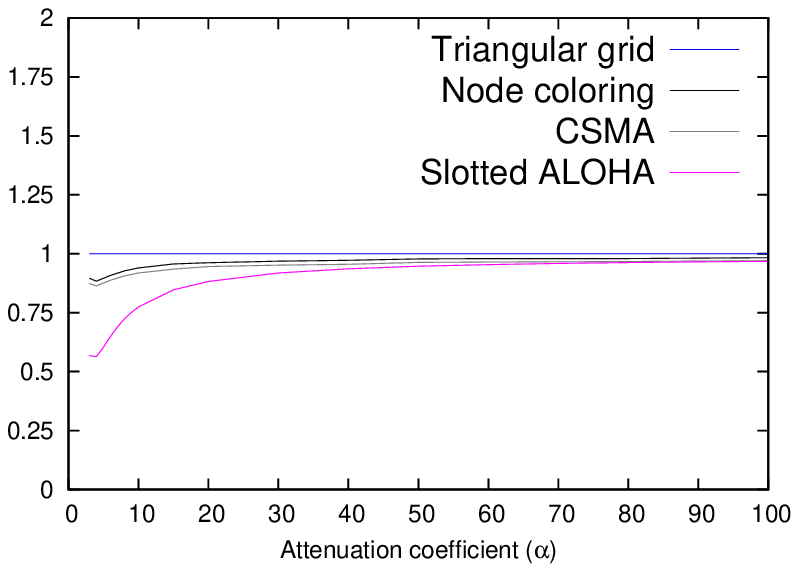}
}
\caption{Scaled comparison of triangular grid pattern, slotted ALOHA, node coloring and CSMA protocols. 
\label{fig:improvement}}
\end{figure*}

\section{Future Work}
\label{sec:future}

In future, we will extend this work to multi-hop networks. A medium access protocol which achieves higher local capacity should also be able to achieve higher end-to-end capacity in multi-hop networks. For example, consider that $\lambda$ is normalized across all protocols to $1$. Therefore, higher local capacity means higher $\sigma(1,K,\alpha)$ which has an impact on the range of transmission and the number of hops required to reach the destination. The analysis to establish exact bounds on end-to-end capacity with different medium access protocols in multi-hop networks will be challenging as we will have to take into account the impact of routing schemes on capacity as well as various parameters like hop length, number of hops and density of simultaneous transmitters which are interrelated. 

The analysis presented here do not take into account fading and shadowing effects. Some results with fading are available, {\it e.g.}, for Poisson distribution of transmitters~\cite{Weber2,Bartek,Haenggi}. Our analysis, in case of slotted ALOHA, can take into account fading by using the results of~\cite{Jacquet:2009}. Nevertheless, analysis of all medium access protocols, discussed here, under the common framework, such as local capacity, is lacking. 

\section{Conclusions}
\label{sec:conclude}

We evaluated the performance of wireless ad hoc networks under the framework of local capacity. Our analysis implies that maximum local capacity in wireless ad hoc networks can be achieved with grid pattern based protocols and our results show that triangular grid pattern outperforms square and hexagonal grids. Moreover, compared to slotted ALOHA, which does not use any significant protocol overhead, triangular grid pattern can only increase the capacity by a factor of $2$ or less whereas CSMA and node coloring can achieve almost similar capacity as the triangular grid pattern based protocol.  

The conclusion of this work is that improvements above ALOHA are limited in performance and may be costly in terms of protocol overheads and that CSMA or node coloring can be very good candidates. Therefore, attention should be focused on optimizing existing medium access protocols and designing efficient routing strategies in case of multi-hop networks. Note that, our results are also relevant when nodes move according to an i.i.d. mobility process such that, at any time, the distribution of nodes in the network is homogeneous.

In future, we will extend this analysis to multi-hop networks and we will also take into account fading and shadowing effects. In case of slotted ALOHA, we can take into account fading by using the results of~\cite{Jacquet:2009} but analysis of all medium access protocols, discussed here, under fading effects and a common framework, such as local capacity, is required.

\bibliographystyle{hieeetr}
\bibliography{local_capacity_arxiv}

\end{document}